\documentclass[letterpaper, 10 pt, conference]{ieeeconf}  

\IEEEoverridecommandlockouts                              

\usepackage{graphicx} 
\usepackage{amsmath} 
\usepackage{amssymb}  

\usepackage{amsthm}  

\usepackage{easy-todo}
\usepackage{algorithm}
\usepackage{algpseudocode}
\algrenewcommand{\algorithmiccomment}[1]{\hskip1em$\%$ #1}
\usepackage{environ} 						
\usepackage{pgfplots}

\usepackage{booktabs}
\usepackage{subcaption}

\usepackage{tablefootnote}
\usepackage{commath} 
\usepackage[short]{optidef}	
\pgfplotsset{compat=newest}
\pgfplotsset{plot coordinates/math parser=false}
\newlength\figureheight
\newlength\figurewidth
\graphicspath{{./figures/}}

\newcommand{\eqd}{ \ .}
\newcommand{\eqc}{ \ ,}
\newcommand{\nt}{M}
\newcommand{\nd}{{n_d}}
\newcommand{\nz}{{n_z}}
\DeclareMathOperator{\E}{\mathbb{E}}
\DeclareMathOperator{\var}{\text{var}}
\DeclareMathOperator{\diag}{\text{diag}}
\newtheorem{Lemma}{Lemma}

\title{\LARGE \bf
Cautious NMPC with Gaussian Process Dynamics for Autonomous Miniature Race Cars
}

\author{Lukas Hewing$^{1}$, {Alexander Liniger$^{2}$} and Melanie N. Zeilinger$^{1}$
\thanks{This work was supported by the Swiss National Science Foundation under grant no. PP00P2 157601 / 1.}
\thanks{$^{1}$Institute for Dynamic Systems and Control, ETH Zurich, Zurich, Switzerland
        {\tt\small lhewing@ethz.ch, mzeilinger@ethz.ch}}%
\thanks{$^{2}$Institute for Automatic Control, ETH Zurich, Zurich, Switzerland
	{\tt\small aliniger@ethz.ch}}%
}

\begin{document}

\maketitle
\thispagestyle{empty}
\pagestyle{empty}

\begin{abstract}
This paper presents an adaptive high performance control method for autonomous miniature race cars. Racing dynamics are notoriously hard to model from first principles, which is addressed by means of a cautious nonlinear model predictive control (NMPC) approach that learns to improve its dynamics model from data and safely increases racing performance. The approach makes use of a Gaussian Process (GP) and takes residual model uncertainty into account through a chance constrained formulation. We present a sparse GP approximation with dynamically adjusting inducing inputs, enabling a real-time implementable controller. The formulation is demonstrated in simulations, which show significant improvement with respect to both lap time and constraint satisfaction compared to an NMPC without model learning.
\end{abstract}

\section{Introduction}
\label{sc:Introduction}
%
Control of autonomous cars is a challenging task and has attracted considerable
attention in recent years~\cite{Buehler2009}. One particular case of autonomous
driving is autonomous racing, where the goal is to drive around a track as fast
as possible, potentially to race against competitors and to avoid
collisions~\cite{Kritayakirana2012}. In order to achieve high performance at
these extreme conditions, racing teams today spend a significant amount of time
and effort on modeling, which is challenging especially near the limits of tire
adhesion~\cite{Guiggiani2014}. Learning-based control methods have been proposed
to address this challenge and show great potential towards improving racing
performance~\cite{Kolter2010}. They do, however, often suffer from poor model
accuracy and performance during transient learning phases. This can lead to
violation of critical constraints~\cite{Akametalu2014} related to keeping the
car on track and avoiding collisions, compromising not only performance, but the
success of the entire race. In addition, iteratively learning the racing task on
a lap-by-lap basis, as considered e.g.\ in~\cite{Kapania2015}, suffers from poor
generalization and does typically not allow for maintaining high performance for
dynamic racing tasks, such as obstacle avoidance or overtaking. This paper
addresses these challenges by learning the dynamics model from data and
considering model uncertainty to ensure constraint satisfaction in a nonlinear
model predictive control (NMPC) approach, offering a flexible framework for
racing control.
\par
Recently, a number of autonomous racing control methods were presented that rely on NMPC formulations. An NMPC racing approach for miniature race cars was proposed in~\cite{Liniger2015}, which uses a contouring control formulation to maximize track progress over a finite horizon and enables obstacle avoidance. It was extended to a stochastic setting in order to take model uncertainty into account in~\cite{Carrau2016} and~\cite{Liniger2017a}. Using model learning in an MPC framework allows for generalizing from collected data and for improving performance in varying racing tasks. This was, for instance, demonstrated in~\cite{Niekerk2017} by using the mean estimate of a Gaussian Process (GP) as a dynamics model for an NMPC method based on~\cite{Liniger2015}. Furthermore, the MPC approach recently proposed in~\cite{Rosolia2017a} was applied to the problem of autonomous racing, where the model is improved with an iterative parameter estimation technique~\cite{Rosolia2017b}.
 \par

The method presented in this paper makes use of GP regression to improve the dynamics model from measurement data, since GPs inherently provide a measure for residual model uncertainty, which is integrated in a cautious NMPC controller.
To this end we extend the approach presented in~\cite{Liniger2015} with a learning module and reformulate the controller in a stochastic setting.
A key element differentiating the approach from available results is the stochastic treatment of a GP model in an NMPC controller to improve both performance and constraint satisfaction properties.
We derive a tractable formulation of the problem that exploits both the improved dynamics model and the uncertainty and show how chance constraints on the states can be approximated in deterministic form.
The framework thereby allows for specifying a minimum probability of satisfying critical constraints, such as track boundaries, offering an intuitive and systematic way of defining a desired trade-off between aggressive driving and safety in terms of collision avoidance. \par
While the use of GPs in MPC offers many benefits, it poses computational challenges for use with fast sampled and larger scale systems, such as the race car problem, since the evaluation complexity of GPs is generally high and directly scales with the number of data points considered. Various approaches to address this limitation have been presented in the literature. One class of methods relies on an approximation by a finite number of basis functions, such as the sparse spectrum approximation~\cite{Lazaro-Gredilla2010}, which is also used in the GP-based NMPC in~\cite{Niekerk2017}. We present an approach for predictive control based on a sparse GP approximation using inducing inputs~\cite{Quinonero-candela2005}, which are selected according to an approximate trajectory in state-action space. This enables a high-fidelity local approximation currently relevant for control at a given measured state, and facilitates real-time implementability of the presented controller.
\par
 We finally evaluate the proposed cautious NMPC controller in simulations of a race. The results demonstrate that it provides safe and high performance control at sampling times of $30 \text{ ms}$, which is computationally on par with NMPC schemes without model learning~\cite{Liniger2015}, while improving racing performance and constraint satisfaction. We furthermore demonstrate robustness towards process noise, indicating fitness for hardware implementation.
%
\section{Preliminaries}
\label{sc:SPGIntro}
%
In the following we specify the notation used in the paper and briefly introduce GP regression and sparse approximations based on inducing inputs as relevant to the presented control approach.
\subsection{Notation}
\label{ssc:Notation}
%
For two matrices or vectors we use $[A;B] := {[A^T \, B^T]}^T$ for vertical matrix/vector concatenation. We use ${[y]}_i$ to refer to the $i$-th element of the vector $y$, and similarly ${[A]}_{\cdot,i}$ for the $i$-th column of matrix A. A normal distribution with mean $\mu$ and variance $\Sigma$ is denoted $\mathcal{N}(\mu,\Sigma)$. We use $\Vert x \Vert$ for the 2-norm of vector $x$ and $\diag(x)$ to express a diagonal matrix with elements given by the vector x. The gradient of a vector-valued function $f: \mathbb{R}^{n_z} \rightarrow \mathbb{R}^{n_f}$ with respect to vector $x~\in~\mathbb{R}^{n_x}$ is denoted $\nabla_{\!x} f : \mathbb{R}^{n_z} \rightarrow  \mathbb{R}^{n_f \times n_x}$.
%
\subsection{Gaussian Process Regression}
\label{ssc:GP}
%
Consider M input locations collected in the matrix $\mathbf{z} = \nobreak [z_1^T; \ldots ; z_\nt^T] \in \mathbb{R}^{\nt \times \nz}$ and corresponding measurements $\mathbf{y} = [y_1^T; \ldots ; y_\nt^T] \in \mathbb{R}^{\nt \times \nd}$ arising from an unknown function $g(z):\mathbb{R}^\nz \rightarrow \mathbb{R}^{\nd}$ under the following statistical model
\begin{equation}
y_j = g(z_j) + \omega_j \eqc \label{eq:likelyhood}
\end{equation}
where $\omega_j$ is i.i.d. Gaussian noise with zero mean and diagonal variance $\Sigma_w = \diag([\sigma^2_1; \ldots; \sigma^2_{\nd}]$. Assuming a GP prior on $g$ in each output dimension $a \in \{1, \ldots, \nd \}$, the measurement data is normally distributed with
\begin{equation*}
	{[\mathbf{y}]}_{\cdot,a} \sim \mathcal{N}(0,K_{\mathbf{z}\mathbf{z}}^a + \sigma_a^2) \eqc
\end{equation*}
where $K_{\mathbf{z}\mathbf{z}}^a$ is the Gram matrix of the data points using the kernel function $k^a(\cdot,\cdot)$ on the input locations $\mathbf{z}$, i.e. ${[K_{\mathbf{z}\mathbf{z}}^a]}_{ij} = k^a(z_i,z_j)$. The choice of kernel functions $k^a$ and its parameterization is the determining factor for the inferred distribution of $g$ and is typically specified using prior process knowledge and optimization based on observed data~\cite{Rasmussen2006}. Throughout this paper we consider the squared exponential kernel function
\begin{equation*}
	k(z,\tilde{z}) = \sigma_f^2 \exp\left(-\frac{1}{2}{(z-\tilde{z})}^T L^{-1} (z-\tilde{z})\right) \eqc
\end{equation*}
in which $L \in \mathbb{R}^{\nz \times \nz}$ is a positive diagonal length scale matrix. It is, however, straightforward to use any other (differentiable) kernel function. \par
The joint distribution of the training data and an arbitrary test point $z$ in output dimension $a$ is given by
\begin{equation}
	p({[y]}_a, {[\mathbf{y}]}_{\cdot,a}) \sim \mathcal{N}\left(0, \begin{bmatrix} K^a_{\mathbf{z}\mathbf{z}} & K^a_{\mathbf{z}z} \\ K^a_{z\mathbf{z}} & K^a_{zz} \end{bmatrix} \right) \label{eq:jointDist} \eqc
\end{equation}
where ${[K^a_{\mathbf{z}z}]}_j = k^a(z_j,z)$, $K^a_{z\mathbf{z}} = {(K^a_{\mathbf{z}z})}^T$ and similarly $K^a_{zz} = k^a(z,z)$. The resulting conditional distribution is Gaussian with $p({[y]}_a\,|\,{[\mathbf{y}]}_{\cdot,a}) \sim \mathcal{N}(\mu^d_a(z),\Sigma^d_a(z))$ and
\begin{subequations}\label{eq:GP_post}
	\begin{align}
		\mu_a^d(z) &=  K_{z\mathbf{z}}^a{(K_{\mathbf{z}\mathbf{z}}^a + I \sigma^2_{a})}^{-1} {[\mathbf{y}]}_{\cdot,a} \eqc \\
		\Sigma_a^d(z) &= K^a_{zz} - K_{z\mathbf{z}}^a{(K_{\mathbf{z}\mathbf{z}}^a + I \sigma^2_{a})}^{-1} K_{\mathbf{z}z}^a \eqd
	\end{align}
\end{subequations}
We call the resulting GP approximation of the unknown function $g(z)$
\begin{equation}
	d(z) \sim \mathcal{N}(\mu^d(z),\Sigma^d(z)) \label{eq:GP}
\end{equation}
with $\mu^d = [\mu^d_1;\ldots;\mu^d_\nd]$ and $\Sigma^d = \diag([\Sigma^d_1;\ldots ;\Sigma^d_\nd])$. \par
Evaluating~\eqref{eq:GP} has cost $\mathcal{O}(\nd \nz \nt)$ and $\mathcal{O}(\nd \nz \nt^2)$ for mean and variance, respectively and thus scales with the number of data points. For many data points or fast real-time applications this limits the use of a GP model. To overcome these issues, various approximation techniques have been proposed, one class of which is sparse Gaussian processes using inducing inputs~\cite{Quinonero-Candela2007}, briefly outlined in the following.
%
\subsection{Sparse Gaussian Processes}
\label{ssc:sGP}
Most sparse GP approximations can be understood using the concept of inducing targets $\mathbf{y}_{ind}$ at inputs $\mathbf{z}_{ind}$ and an inducing conditional distribution $q$ to approximate the joint distribution~\eqref{eq:jointDist} by assuming that test points and training data are conditionally independent given $\mathbf{y}_{ind}$~\cite{Quinonero-candela2005}:
\begin{align*}
	&p({[y]}_a,{[\mathbf{y}]}_{\cdot,a}) = \int \! p({[y]}_a,{[\mathbf{y}]}_{\cdot,a} \, | \, \mathbf{y}_{ind}) p(\mathbf{y}_{ind}) \, \dif \mathbf{y}_{ind} \\&\approx 	\int \! q({[y]}_a \, | \, \mathbf{y}_{ind}) q({[\mathbf{y}]}_{\cdot,a} \, | \, \mathbf{y}_{ind}) p(\mathbf{y}_{ind}) \dif \mathbf{y}_{ind} \, .
\end{align*}
There are numerous options for selecting the inducing inputs, e.g.\ heuristically as a subset of the original data points, by treating them as hyperparameters and optimizing their location~\cite{Snelson2006}, or letting them coincide with test points~\cite{Tresp2000}. \par
In this paper, we make use of the state-of-the-art Fully Independent Training Conditional (FITC) approximation to approximate the GP distribution and reduce computational complexity~\cite{Snelson2006}. Given a selection of inducing inputs $\mathbf{z}_{ind}$ and using the shorthand notation $Q^a_{\zeta\tilde{\zeta}} := \nobreak K^a_{\zeta\mathbf{z}_{ind}} {(K^a_{\mathbf{z}_{ind}\mathbf{z}_{ind}})}^{-1} K^a_{\mathbf{z}_{ind}\tilde{\zeta}}$ the approximate posterior distribution is given by
\begin{subequations}\label{eq:sGP}
	\begin{align}
		\tilde{\mu}^{d}_a(z) &= Q^a_{z\mathbf{z}}{(Q^a_{\mathbf{z}\mathbf{z}} + \Lambda)}^{-1}{[\mathbf{y}]}_{\cdot,a} \eqc \\
		\tilde{\Sigma}^{d}_a(z) & = K^a_{zz} - Q^a_{z\mathbf{z}}{(Q^a_{\mathbf{z}\mathbf{z}} + \Lambda)}^{-1}Q_{\mathbf{z}z}
	\end{align}
\end{subequations}
with $\Lambda = \diag(K^a_{\mathbf{z}\mathbf{z}} - Q^a_{\mathbf{z}\mathbf{z}} + I \sigma^2_{a})$. Concatenating the output dimensions similar to~\eqref{eq:GP} we arrive at the approximation
\begin{equation*}
	\tilde{d}(z) \sim \mathcal{N}(\tilde{\mu}^d(z),\tilde{\Sigma}^d(z)) \eqd
\end{equation*}
Several of the matrices used in~\eqref{eq:sGP} can be precomputed such that the evaluation complexity becomes independent of the number of original data points. Using $\tilde{M}$ inducing points, the computational complexity for evaluating the sparse GP at a test point is reduced to $\mathcal{O}(\nd \nz \tilde{M})$ and $\mathcal{O}(\nd \nz \tilde{M}^2)$ for the predictive mean and variance, respectively. 
%
%
\section{Race Car Modeling}
\label{sc:system}
%
%
\begin{figure}[h]
	\centering
	\def\svgwidth{5.5cm}
	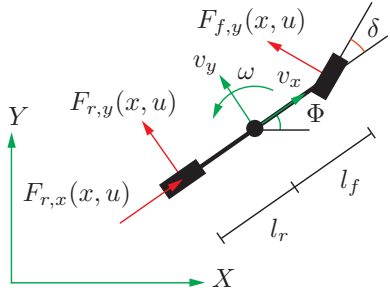
	\caption{Schematic of the car model.}\label{fig:bicycle}
\end{figure}
This section presents the race car setup and nominal modeling of the car dynamics, which will serve as a base model for the learning-based control approach. This is largely based on material presented in~\cite{Liniger2015}, which provides a more detailed exposition.
\subsection{Car Dynamics}
\label{ssc:NominalDynamics}
%
We consider the following model structure to describe the dynamics of the miniature race cars
\begin{align} \label{eq:contSys}
	\dot{x} =f_c(x,u) + B_d (g_c(x,u) + w) \eqc
\end{align}
where $f_c(x,u)$ are the nominal system dynamics of the car modeled from first principles, and $g_c(x,u)$ reflects unmodeled dynamics. The considered nominal dynamics are obtained from a bicycle model with nonlinear tire forces as shown in Figure~\ref{fig:bicycle}, resulting in
\begin{align}
f_c(x,u) & = \begin{bmatrix*}[l] v_x \cos(\Phi) - v_y \sin(\Phi)\\
v_x \sin(\Phi) + v_y \cos(\Phi)\\
\omega \\
\frac{1}{m}\Big(F_{r,x}(x,u) - F_{f,y}(x,u) \sin{\delta} + m v_y \omega \Big)\\
\frac{1}{m} \Big( F_{r,y}(x,u) + F_{f,y}(x,u) \cos{\delta} - m v_x \omega \Big) \\
\frac{1}{I_z}\Big(F_{f,y}(x,u) l_f \cos{\delta} - F_{r,y}(x,u) l_r \Big)\end{bmatrix*} \eqc \label{eq:nomDynamics_cont}
\end{align}
where $x = [X;Y;\Phi;v_x;v_y;\omega]$ is the state of the system, with position $(X,Y)$, orientation $\Phi$, longitudinal and lateral velocities $v_x$ and $v_y$, and yaw rate $\omega$. The inputs to the system are the motor duty cycle $p$ and the steering angle $\delta$, i.e., $u=[p; \delta]$. Furthermore, $m$ is the mass, $I_z$ the moment of inertia and $l_r$ and $l_f$ are the distance of the center of gravity from the rear and front tire, respectively. The most difficult components to model are the tire forces $F_{f,y}$ and $F_{r,y}$ and the drivetrain force $F_{r,x}$. The tires are modeled by a simplified Pacejka tire model~\cite{Pacejka1992} and the drivetrain using a DC motor model combined with a friction model. For the exact formulations of the forces, we refer to~\cite{Liniger2015}. \par
In order to account for model mismatch due to inaccurate parameter choices and limited fidelity of this simple model, we integrate $g_c(x,u)$ capturing unmodeled dynamics, as well as additive Gaussian white noise $w$.
Due to the structure of the nominal model, i.e.\ since the dynamics of the first three states are given purely by kinematic relationships, we assume that the model uncertainty, as well as the process noise $w$, only affect the velocity states $v_x$, $v_y$ and $\omega$ of the system, that is $B_d = [0;I_3]$. \par
For the use in a discrete-time MPC formulation, we finally discretize the system using the Euler forward scheme with a sampling time of $T_s$, resulting in the following description,
\begin{align}
	x(k\!+\!1) = f(x(k),u(k)) \! + \! B_d (g(x(k),u(k)) \! + \! w(k)), \label{eq:disc_system}
\end{align}
where $w(k)$ is i.i.d.\ normally distributed process noise with $w(k) \sim \mathcal{N}(0,\Sigma^w)$ and $\Sigma^w = \text{diag}[\sigma^2_{v_x};\sigma^2_{v_y};\sigma^2_{\omega}]$, which, together with the uncertain dynamics function $g$, will be inferred from measurement data.
%
\subsection{Race Track and Constraints}
\label{ssc:RaceTrack}
%
We consider a race track given by its centerline and a fixed track width. The centerline is described by a piecewise cubic spline polynomial, which is parametrized by the path length $\Theta$. Given a $\Theta$, we can evaluate the corresponding centerline position $(X_c(\Theta),Y_c(\Theta))$ and orientation $\Phi_c(\Theta)$. By letting $\tilde{\Theta}$ correspond to the projection of $(X,Y)$ on the centerline, the constraint for the car to stay within the track boundaries is expressed as
\begin{equation}
	\mathcal{X}(\tilde{\Theta}) := \left\{x \, \middle| \, \left \Vert \begin{bmatrix} X \\ Y \end{bmatrix} - \begin{bmatrix} X_c(\tilde{\Theta}) \\ Y_c(\tilde{\Theta}) \end{bmatrix} \right\Vert \leq r \right\} \eqc \label{eq:trackCon}
\end{equation}
where $r$ is half the track width. \par
Additionally, the system is subject to input constraints,
\begin{align}
	\mathcal{U} = \left\{u \, \middle| \, \begin{bmatrix} 0 \\ -\delta_{\max} \end{bmatrix} \leq \begin{bmatrix} p \\ \delta \end{bmatrix} \leq \begin{bmatrix} 1 \\ \delta_{\max} \end{bmatrix}  \right\} \eqc \label{eq:uCon}
\end{align}
i.e.\ the steering angle is limited to a maximal angle $\delta_{\max}$ and the duty cycle has to lie between zero and one.
%
\section{Learning-based Controller Design}
\label{sc:ControllerDesign}
%
In the following, we first present the model learning module that is subsequently used in a cautious NMPC controller. We briefly state the contouring control formulation~\cite{Liniger2015}, serving as the basis for the controller and integrate the learning-based dynamics using a stochastic GP model. Afterwards, we introduce suitable approximations to reduce computational complexity and render the control approach real-time feasible.
%
\subsection{Model Learning}
\label{ssc:ModelLearning}
%
We apply Gaussian process regression~\cite{Rasmussen2006} to infer the vector-valued function $g$ of the discrete-time system dynamics~\eqref{eq:disc_system} from previously collected measurement data of states and inputs. Training data is generated as the deviation to the nominal system model, i.e.\ for a specific data point:
\begin{align*}
	y_j &= g(x(j),u(j)) + w(j) = B_d^{\dagger}\left(x(j\!+\!1) - f(x(j),u(j))\right) \eqc \\
	z_j &= [x(j) ; u(j)] \eqc
\end{align*} %
where $^\dagger$ is the pseudoinverse. Note that this is in the form of~\eqref{eq:likelyhood} and we can directly apply~\eqref{eq:GP_post} to derive a GP model $d(x_i,u_i)$ from the data, resulting in the stochastic model
\begin{align}
	x_{i+1} = f(x_{i},u_i) + B_d (d(x_i,u_i) + w_i) \eqd \label{eq:learned_system}
\end{align}
The state $x_{i}$ obtained from this model, which will be used in a predictive controller, is given in form of a stochastic distribution.
%
\subsection{Contouring Control}
\label{ssc:ContouringControl}
%
The learning-based NMPC controller makes use of a contouring control formulation, which has been introduced in~\cite{Faulwasser2009, Lam2010} and was shown to provide good racing performance in~\cite{Liniger2015}.
The objective of the optimal contouring control formulation is to maximize progress along the race track. An approximation of the car position along the centerline is introduced as an optimization variable by including integrator dynamics $\Theta_{i+1} = \Theta_{i} + v_i$, where $\Theta_i$ is a position along the track at time step $i$ and $v_i$ is the incremental progress. The progress along the centerline over the horizon is then maximized by means of the overall incremental progress $\sum_{i = 0}^N v_i$. \par
In order to connect the progress variable to the race car's position, $\Theta_i$ is linked to the projection of the car on the centerline. This is achieved by minimizing the so-called lag error $\hat{e}^l$ and contouring error $\hat{e}^c$, defined as
\begin{align*}
    \hat{e}^l(x_i,\Theta_i) = &-\cos(\Phi(\Theta_i))(X_i-X_c(\Theta_i)) \nonumber \\  &- \sin(\Phi(\Theta_i))(Y_i-Y_c(\Theta_i))\,,\\
        \hat{e}^c(x_i,\Theta_i) = &\phantom{+}\sin(\Phi(\Theta_i))(X_i-X_c(\Theta_i)) \nonumber \\ &- \cos(\Phi(\Theta_i))(Y_i-Y_c(\Theta_i))\,.
\end{align*}
For small contouring error $\hat{e}^c$, the lag error $\hat{e}^l$ approximates the distance between the projection of the car's position and $(X_c(\Theta_i),Y_c(\Theta_i))$, such that a small lag error ensures a good approximate projection. The stage cost function is then formulated as
\begin{align}
    l(x_i,u_i,\Theta_i,v_i) = &\Vert \hat{e}^c(x_i,\Theta_i) \Vert^2_{q_c} + \Vert \hat{e}^l(x_i,\Theta_i) \Vert^2_{q_l} \nonumber \\
        &- \gamma v_i + l_{reg}(\Delta u_i, \Delta v_i) \eqd \label{eq:cost}
\end{align}
The term $-\gamma v_i$ encourages the progress along the track, using the relative weighting parameter $\gamma$. The parameters $q_c$ and $q_l$ are weights on contouring and lag error, respectively, and $l_{reg}(\Delta u_i,\Delta v_i)$ is a regularization term penalizing large changes in the control input and incremental progress $l_{reg}(\Delta u_i,\Delta v_i) = \Vert u_i - u_{i-1} \Vert^2_{R_u} + \Vert v_i - v_{i-1} \Vert^2_{R_v}$, with the corresponding weights $R_u$ and $R_v$.
\par
Based on this contouring formulation, we define a stochastic MPC problem that integrates the learned GP-model~\eqref{eq:learned_system} and minimizes the expected value of the cost function~\eqref{eq:cost} over a finite horizon of length N\@:
\begin{mini!}
    {U,V}{\E\left(\sum_{i=0}^{N-1} l(x_i,u_i,\Theta_i,v_i) \right)\label{eq:costExpValue}}
    {\label{eq:opt_orig}}{}
    \addConstraint{x_{i+1}}{= f(x_{i},u_{i}) + B_d (d(x_{i},u_{i}) + w_i)}
    \addConstraint{\Theta_{i+1}}{ = \Theta_i + v_i}
    \addConstraint{P(x_{i+1}}{ \in \mathcal{X}(\Theta_{i+1}))> 1- \epsilon\label{eq:ChanceState}}
    \addConstraint{u_{i}}{\in \mathcal{U}}
    \addConstraint{x_{0}}{= x(k) , \, \Theta_0 = \Theta(k) \eqc}
\end{mini!}
where $i=0,\ldots,N\!-\!1$ and $x(k)$ and $\Theta(k)$ are the current system state and the corresponding position on the centerline. The state constraints are formulated w.r.t.\ the centerline position at $\Theta_i$ as an approximation of the projection of the car position, and are in the form of chance constraints which guarantee that the track constraint~\eqref{eq:trackCon} is violated with a probability less than $1-\epsilon$. \par
Solving problem~\eqref{eq:opt_orig} is computationally demanding, especially since the distribution of the state is generally not Gaussian after the first prediction time step. In addition, fast sampling times -- in the considered race car setting of about $30 \text{ ms}$ -- pose a significant challenge for real-time computation. In the following subsections, we present a sequence of approximations to reduce the computational complexity of the GP-based NMPC problem for autonomous racing in~\eqref{eq:opt_orig} and eventually provide a real-time feasible approximate controller that can still leverage the key benefits of learning.
\subsection{Approximate Uncertainty Propagation}
\label{ssc:UncertProp}
%
At each time step, the GP $d(x_i,u_i)$ evaluates to a stochastic distribution according to the residual model uncertainty, which is then propagated forward in time, rendering the state distributions non-Gaussian. In order to solve~\eqref{eq:opt_orig}, we therefore approximate the distributions of the state at each prediction step as a Gaussian, i.e. $x_i \sim \mathcal{N}(\mu^x_i, \Sigma^x_i)$~\cite{Candela2003,Deisenroth2010,Hewing2017}. The dynamics equations for the Gaussian distributions can be found e.g.\ through a sigma point transform~\cite{Ostafew2016} or a first order Taylor expansion detailed in Appendix~\ref{app:EKF_GP}. We make use of the Taylor approximation offering a computationally cheap procedure of sufficient accuracy, resulting in the following dynamics for the mean and variance
\begin{subequations}
\begin{align}
	\mu^x_{i+1} &= f(\mu^x_i,u_i) + B_d \mu^d(\mu^x_i,u_i) \eqc \\
	\Sigma^x_{i+1} &= \tilde{A}_i \begin{bmatrix} \Sigma^x_i & \star \\ \nabla_{\!x} \mu^d(\mu^x_i,u_i) \Sigma^x_i & \Sigma^d(\mu^x_i,u_i) \end{bmatrix} \tilde{A}_i^T \label{eq:var_dyn} \eqc
\end{align}
\end{subequations}
where $\tilde{A}_i = \begin{bmatrix} \nabla_{\!x} f(\mu^x_i,u_i) & B_d \end{bmatrix}$ and the star denotes the corresponding element of the symmetric matrix. \par
%
\subsection{Simplified Chance Constraints}
\label{ssc:chance_constrains}
%

The Gaussian approximation of the state distribution allows for a simplified treatment of the chance constraints~\eqref{eq:ChanceState}. They can be approximated as deterministic constraints on mean and variance of the state using the following Lemma.
\begin{Lemma}\label{lm:chanceConstr}
	Let $n$-dimensional random vector $x \sim \mathcal{N}(\mu,\Sigma)$ and the set $\mathcal{B}^{x_c}(r) =\nobreak \left\{ x \, | \, \Vert x-x_c \Vert \leq r \right\}$. Then
	\[\Vert \mu - x_c \Vert \leq r - \sqrt{\chi^2_n(p) \lambda_{max}(\Sigma)} \Rightarrow \Pr(x \in \mathcal{B}^{x_c}(r)) \geq p,\]
	where $\chi^2_n(p)$ is the quantile function of the chi-squared distribution with $n$ degrees of freedom and $\lambda_{max}(\Sigma)$ the maximum eigenvalue of $\Sigma$.
\end{Lemma}
\begin{proof} Let $\mathcal{E}^x_p := \{ x \, | \, {(x-\mu)}^T \Sigma^{-1} (x-\mu) \leq \chi^2_n(p) \}$ be the confidence region of $x$ at level $p$, such that $\Pr(x \in \nobreak \mathcal{E}^x_p) \geq \nobreak p$.
  We have $\mathcal{E}^x_p \subseteq \mathcal{E}^{\tilde{x}}_p$ with $\tilde{x} \sim \mathcal{N}(\mu,\lambda_{max}(\Sigma) \, I)$, i.e. $\mathcal{E}^{\tilde{x}}_p$ is an outer approximation of the confidence region using the direction of largest variance.
  Now $\mu \in  \mathcal{B}^{x_c}(r- \nobreak \sqrt{\chi^2_n(p) \lambda_{max}(\Sigma)})$ implies $\mathcal{E}^{\tilde{x}}_p \subseteq \nobreak \mathcal{B}^{x}_c(r)$, which means $\Pr(x \in \mathcal{B}^{x}_c(r)) \geq \Pr(x \in \nobreak \mathcal{E}^{\tilde{x}}_p) \geq \Pr(x \in \nobreak \mathcal{E}^{x}_p) =\nobreak p$.
\end{proof}
Using Lemma~\ref{lm:chanceConstr}, we can formulate a bound on the probability of track constraint violation by enforcing
\begin{align} \label{eq:ChanceConstrMean}
\left\Vert \begin{bmatrix} \mu^X_i \\ \mu^Y_i \end{bmatrix} - \begin{bmatrix} X_c(\Theta_i) \\ Y_c(\Theta_i) \end{bmatrix} \right\Vert \leq r - \sqrt{\chi^2_2(p) \lambda_{max}(\Sigma_i^{XY})},
\end{align}
where $\Sigma_i^{XY} \in \mathbb{R}^{2 \times 2}$ is the marginal variance of the joint distribution of $X_i$ and $Y_i$.
This procedure is similar to constraint tightening in robust control. Here the amount of tightening is related to an approximate confidence region for the deviation from the mean system state.\par
Constraint $\eqref{eq:ChanceConstrMean}$ as well as the cost~\eqref{eq:cost} require the variance dynamics. The next section proposes a further simplification to reduce computational cost by considering an approximate evolution of the state variance.
%
\subsection{Time-Varying Approximation of Variance Dynamics}
\label{ssc:TVapprox}
%
The variance dynamics in~\eqref{eq:var_dyn} require $\frac{N}{2}(n^2 + n)$ additional variables in the optimization problem and can increase computation time drastically. We trade off accuracy in the system description with computational complexity by evaluating the system variance around an approximate evolution of the state and input. This state-action trajectory can typically be chosen as a reference to be tracked or by shifting a solution of the MPC optimization problem at an earlier time step. Denoting a point on the approximate state-action trajectory with $(\bar{\mu}^x_i, \bar{u}_i)$, the approximate variance dynamics are given by
\begin{align*}
\bar{\Sigma}^x_{i+1} &= \bar{A}_i \begin{bmatrix} \bar{\Sigma}^x_i & \star \\ \nabla_{\!x} \mu^d(\bar{\mu}^x_i,\bar{u}_i) \bar{\Sigma}^x_i & \Sigma^d(\bar{\mu}^x_i,\bar{u}_i) \end{bmatrix} \bar{A}_i^T
\end{align*}
with $\bar{A}_i = [\nabla_{\!x} f(\bar{\mu}^x_i, \bar{u}_i)\ B_d]$. The variance along the trajectory thus does not depend on any optimization variable and can be computed before the state measurement becomes available at each sampling time. The precomputed variance is then used to satisfy the chance constraints approximately, by replacing $\Sigma^{XY}$ with $\bar{\Sigma}^{XY}$ in~\eqref{eq:ChanceConstrMean}. The resulting set is denoted $\bar{\mathcal{X}}(\bar{\Sigma}^x_i, \Theta_i)$.
Figure~\ref{fg:uncert_ex} shows an example of a planned trajectory with active chance constraints according to this formulation with $\chi^2_2(p) = 1$. \par
\begin{figure}
	\center
	\includegraphics[width=0.54\linewidth]{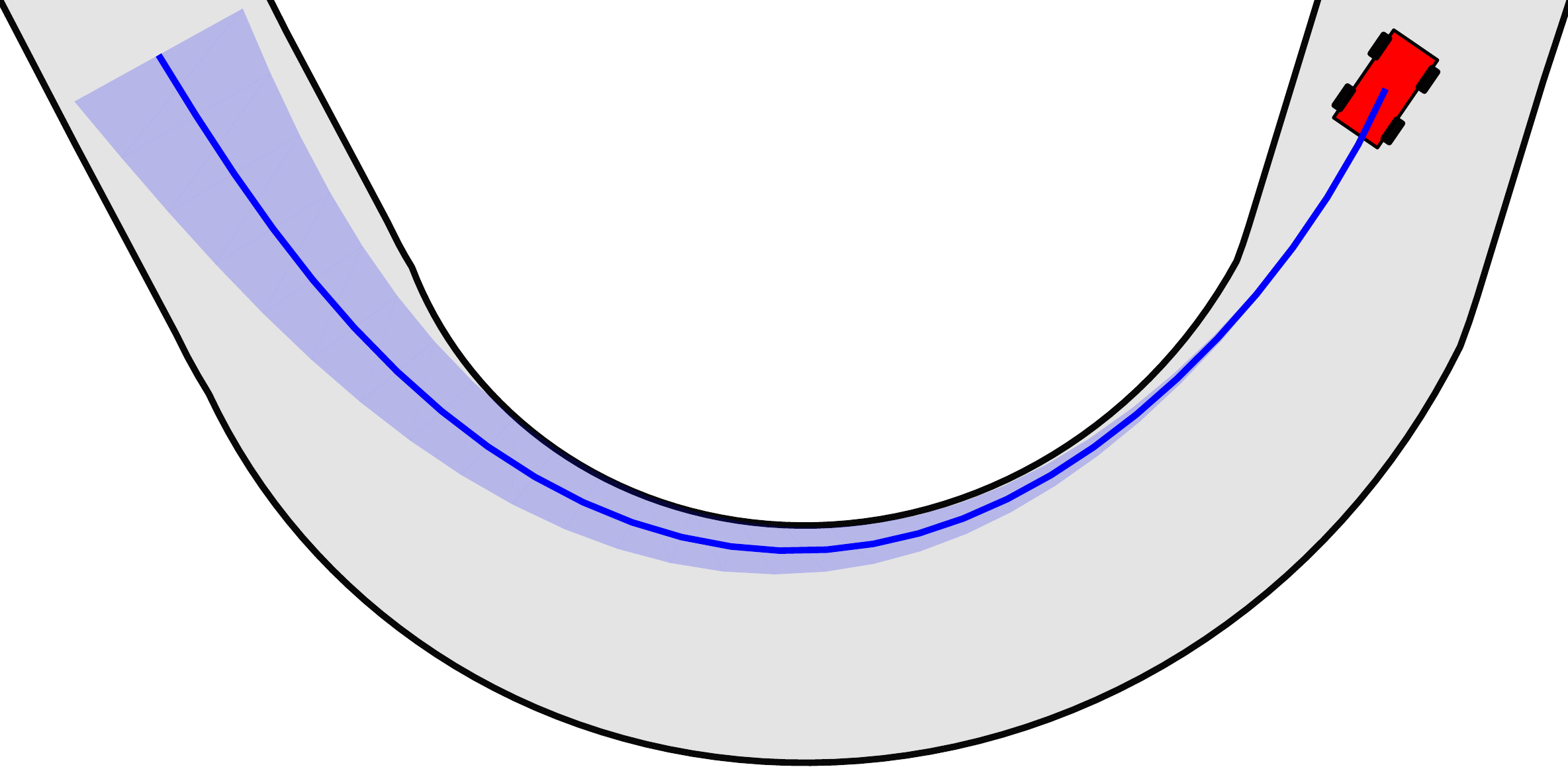}
	\caption{Planned trajectory with active chance constraints. Shown is the mean trajectory of the car with 1-$\sigma$ confidence level perpendicular to the car's mean orientation.}\label{fg:uncert_ex}
\end{figure}
In the following, we use similar ideas to reduce the computational complexity of the required GP evaluations by dynamically choosing inducing inputs in a sparse GP approximation.
%
\subsection{Dynamic Sparse GP}
\label{ssc:DSGP}
%
Sparse approximations as outlined in Section~\ref{ssc:sGP} can considerably speed up evaluation of a GP, with little deterioration of prediction quality. For fast applications with high-dimensional state-input spaces, however, the computational burden can still be prohibitive.
\par
We therefore propose to select inducing inputs locally at each sampling time, which relies on the idea that in MPC the area of interest at each sampling time typically lies close to a known trajectory in the state-action space. Similar to the approximation presented in the previous subsection, inducing inputs can then be selected along the approximate trajectory, e.g.\ according to a solution computed at a previous time step.
\par
We illustrate the procedure using a two-dimensional example in Figure~\ref{fg:DSGP_toyEx} showing the dynamic approximation for a simple double integrator. Shown is the contour plot of the posterior variance of a GP with two input dimensions $x_1$ and $x_2$. Additionally, two trajectories generated from an MPC are shown. The solid red line corresponds to a current prediction trajectory, while the dashed line shows the previous prediction, which is used for local approximation of the GP\@. As the figure illustrates, full GP and sparse approximation are in close correspondence along the predicted trajectory of the system. \par
The dynamic selection of local inducing points in a receding horizon fashion allows for an additional speed-up by computing successive approximations adding or removing single inducing points by means of rank 1 updates~\cite{Seeger2004}. These are applied to a reformulation of~\eqref{eq:sGP}, which offers better numerical properties~\cite{Quinonero-candela2005} and avoids inversion of the large matrix $Q^a_{\mathbf{z}\mathbf{z}} + \Lambda$,
\begin{subequations}
	\begin{align*}
		\tilde{\mu}^{a}_d(z) &= K^a_{z\mathbf{z}} \Sigma K^a_{\mathbf{z}_{ind},\mathbf{z}} \Lambda^{-1}{[\mathbf{y}]}_{\cdot,a} \, ,\\
		\tilde{\Sigma}^{a}_d(z) & = K^a_{zz} - Q^a_{zz} + K^a_{z\mathbf{z}_{ind}} \Sigma K^a_{\mathbf{z}_{ind} z} \eqc
	\end{align*}
\end{subequations}
with $\Sigma = {\left(K^a_{\mathbf{z}_{ind}\mathbf{z}_{ind}} + K^a_{\mathbf{z}_{ind} \mathbf{z}} \Lambda^{-1} K^a_{\mathbf{z} \mathbf{z}_{ind}}\right)}^{-1}$. Substitution of single inducing points corresponds to a single line and column changing in $\Sigma^{-1}$. The corresponding Cholesky factorizations can thus efficiently be updated~\cite{Nguyen-Tuong2009}.
\begin{figure}
	\center
	\setlength\figureheight{6cm}
	\setlength\figurewidth{7.5cm}
	\input{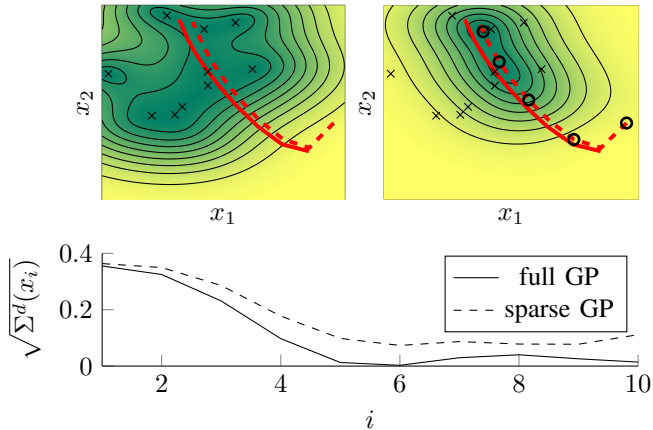}
	\vspace{-0.6cm}
	\caption{Contour plots of the posterior variance of a GP for the full GP (top left) and dynamic sparse approximation (top right). The  solid red line is the trajectory planned by an MPC, the dashed red line the trajectory of the previous time step used for the approximation, with inducing points indicated by black circles. The bottom plot shows the respective variances along the planned trajectory.}\label{fg:DSGP_toyEx}
\end{figure}
%
%
\subsection{Resulting Control Formulation for Autonomous Racing}
\label{ssc:Resulting Formulation}
%
We integrate the approximations presented in the previous sections in the learning-based MPC problem in~\eqref{eq:opt_orig} resulting in the following approximate optimization problem
\begin{mini!}
	{U,V}{ \E \left(\sum_{i=0}^{N-1} l(\mu^x_i,u_i,\Theta_i,v_i) \right)}
	{\label{eq:final_formulation}}{}
	\addConstraint{\mu^x_{i+1} }{= f(\mu^x_{i},u_{i}) + B_d \tilde{\mu}^d(\mu^x_{i},u_{i})}
	\addConstraint{\Theta_{i+1}}{ = \Theta_i + v_i}
	\addConstraint{\mu^x_{i+1}}{ \in \bar{\mathcal{X}}(\bar{\Sigma}^x_{i+1},\Theta_{i+1})\label{eq:chanceConstrFinal}}
	\addConstraint{u_i }{\in \mathcal{U}}
	\addConstraint{\mu^x_{0}}{ = x(k), \, \Theta_0 = \Theta(k) \, ,}
\end{mini!}
where $i = 0,\ldots,N\!-\!1$. By reducing the learned model to the mean GP dynamics and considering approximate variance dynamics and simplified chance constraints, the problem is reduced to a deterministic nonlinear program of moderate dimension. \par
In the presented form, the approximate optimization problem~\eqref{eq:final_formulation} still requires an optimization over a large spline polynomial corresponding to the entire track.
Since evaluation of this polynomial and its derivative is computationally expensive, one can apply an additional approximation step and quadratically approximate the cost function around the shifted solution trajectory from the previous sampling time, for which the expected value is equivalent to the cost at the mean. Similarly, $\Theta_i$ can be fixed using the previous solution when evaluating the state constraints~\eqref{eq:chanceConstrFinal}, such that the spline can be evaluated separately from the optimization procedure, as done in~\cite{Liniger2015}.
%
\section{Simulation}
\label{sc:simulation}
%
We finally evaluate the proposed control approach in simulations of a race. The race car is simulated using system~\eqref{eq:contSys} with $g_c$ resulting from a random perturbation of all parameters of the nominal dynamics $f_c$ by up to $\pm 15\%$ of their original value. We compare two GP-based approaches, one using the full GP $d(x_i,u_i)$ with all available data points and one a dynamic sparse approximation $\tilde{d}(x_i,u_i)$, against a baseline NMPC controller, which makes use of only the nominal part of the model $f_c$, as well as against a reference controller using the true system model, i.e.\ with knowledge of $g_c$.
%
\subsection{Simulation Setup}
\label{sc:simulation_setup}
%
We generate controllers using formulation~\eqref{eq:final_formulation}, both for the full GP and the dynamic sparse approximation with 10 inducing inputs along the previous solution trajectory of the MPC problem. The inducing points are placed with exponentially decaying density along the previous solution trajectory, putting additional emphasis on the current and near future states of the car. The prediction horizon is chosen as $N = 30$ and we formulate the chance constraints~\eqref{eq:chanceConstrFinal} with $\chi^2_2(p) = 1$.
To guarantee feasibility of the optimization problem, we implement the chance constraint using a linear quadratic soft constraint formulation. Specifically, we use slack variables $s_i \geq 0$, which incur additional costs $l_s(s_i) = \Vert s_i \Vert^2_{q_s} + c_s s_i$. For sufficiently large $c_s$ the soft constrained formulation is exact, if feasible~\cite{Kerrigan2000}.
To reduce conservatism of the controllers, constraints are only tightened for the first 15 prediction steps and are applied to the mean for the remainder of the prediction horizon, similar to the method used in~\cite{Carrau2016}.\par
The system is simulated for one lap of a race, starting with zero initial velocity from a point on the centerline under white noise of power spectral density $Q_w = \frac{1}{T_s}\diag([0.001;0.001;0.1])$. The resulting measurements from one lap with the baseline controller are used to generate 350 data-points for both GP-based controllers. Hyperparameters and process noise level were found through likelihood optimization, see e.g.~\cite{Rasmussen2006}. \par
To exemplify the learned deviations from the nominal system, Figure~\ref{fg:GPprediction} shows the encountered dynamics error in the yaw-rate and the predicted error during a lap with the sparse GP-based controller. Overall, the learned dynamics are in good correspondence with the true model and the uncertainty predicted by the GP matches the residual model uncertainty and process noise well. Note that the apparent volatility in the plot does not correspond to overfitting, but instead is due to fast changes in the input and matches the validation data.\par
Solvers were generated using FORCES Pro~\cite{FORCESPro} with a sampling time of $T_s = 30 \text{ ms}$ and the number of maximum solver iterations were limited to 75, which is sufficient to guarantee a solution of required accuracy. All simulations were carried out on a laptop computer with a 2.6 GHz i7-5600 CPU and 12GB RAM\@. \par
\begin{figure}
	\center
	\setlength\figureheight{3cm}
	\setlength\figurewidth{7cm}
	\input{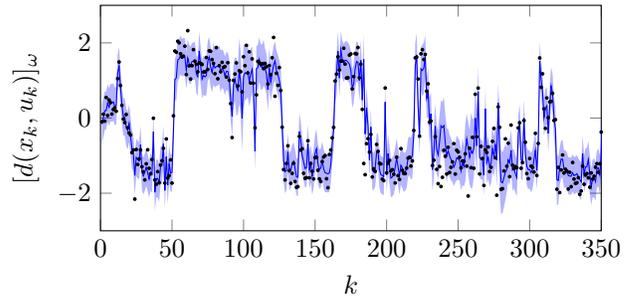}
	\caption{Prediction of the dynamic sparse GP with 10 inducing inputs during a race lap. Shown as black dots are the error on the yaw rate under process noise as encountered at each time step. The blue line shows the dynamics error predicted by the GP \@. The shaded region indicates the 2-$\sigma$ confidence interval, including noise.}\label{fg:GPprediction}
\end{figure}
%
%
\subsection{Results}
\label{ssc:results}
%
%
To quantify performance of the proposed controllers we compare the lap time $T_l$ and the average squared slack of the realized states $\overline{s_0^2}$ corresponding to state-constraint violations. We furthermore state average solve times $\overline{T}_{\!c}$ of the NMPC problem and its $99.9$th percentile $T_c^{99.9}$ over the simulation run. To demonstrate the learning performance we also evaluate the average 2-norm error in the system dynamics $\overline{\Vert e \Vert}$, i.e.\ the difference between the mean state after one prediction step and the realized state, $e(k\!+\!1) = \mu^x_1 - x(k\!+\!1)$. \par
\begin{figure}
	\center
	\setlength\figureheight{7cm}
	\setlength\figurewidth{6.5cm}
	\input{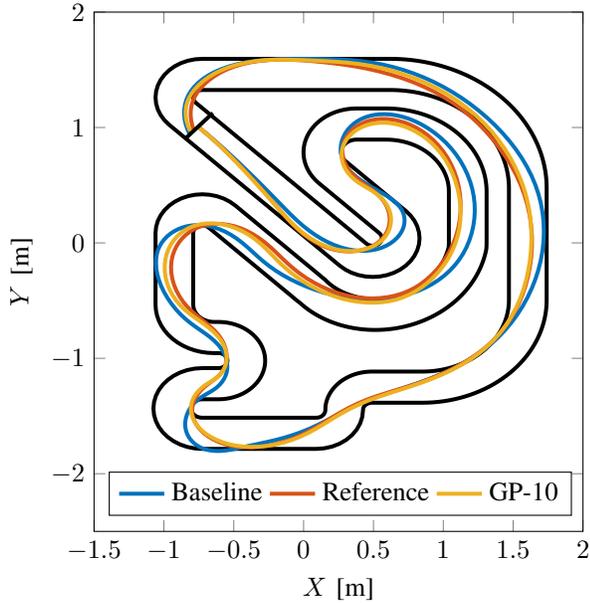}
	\caption{Resulting trajectories on the race track for simulations without process noise with baseline, reference and sparse GP-based controller.} \label{fg:res_no_noise}
\end{figure}
For direct comparison, we first evaluate controller performance in simulations without process noise. As evident in Figure~\ref{fg:res_no_noise}, the baseline controller performs visually suboptimally and is unable to guarantee constraint satisfaction, even in the absence of process noise. The reference controller and sparse GP-based controller (GP-10) perform similarly. Table~\ref{tb:res}(a) summarizes the results of the simulations without process noise. We can see that the full GP controller (GP-Full) matches the performance of the reference controller. It also displays only small constraint violations, while the reference controller exhibits some corner cutting behavior leading to constraint violations. This is due to unmodeled discretization error, also evident in the dynamics error of the reference controller. The discretization error is partly learned by the GPs, leading to lower error than even the reference controller. Overall the sparse GP controller demonstrates a performance close to that of the full GP controller, both in terms of lap time and constraint satisfaction and is able to significantly outperform the baseline controller.\par
\begin{table}[h]
	\center
	\begin{minipage}{\linewidth} \caption{Simulation results}\label{tb:res}
		\begin{subtable}{\linewidth} \caption{without process noise \hspace*{\fill}}  \vspace{-0.1cm}
			\begin{tabular}{cccccc} \toprule
				Controller	& $T_l$ [s] & $\overline{s_0^2}$ [$10^{-3}$] & $\overline{\Vert e \Vert}$ [-] & $\overline{T}_{\!c}$ [ms] & $T_c^{99.9}$ [ms]\\  \midrule
				Reference	& 8.64 & 4.50 & 0.18 & 9.4 & 19.1\\
				Baseline 	& 9.45 & 4.77 & 1.20  & 10.8 & 20.6\\
				GP-Full 	& 8.67 & 0.95 & 0.09 & 105.2 & 199.23\\
				GP-10$^a$	& 8.76 & 1.77 & 0.16 & 12.3	 & 26.9
			\end{tabular}
			\vspace{0.3cm}
		\end{subtable}
		\begin{subtable}{\linewidth} \caption{with process noise \hspace*{\fill}} \vspace{-0.1cm}\label{tb:res_noise}
			\begin{tabular}{cccccc} \toprule
				Controller	& $T_l$ [s] & $\overline{s_0^2}$ [$10^{-3}$] & $\overline{\Vert e \Vert}$ [-] & $\overline{T}_{\!c}$ [ms] & $T_c^{99.9}$ [ms] \\  \midrule
				Reference		& 8.76 & 2.88 & 0.33 & 9.7 & 20.8 \\
				Baseline$^b$ 	& 9.55 & 65.11 & 1.20  & 10.1 & 23.9\\
				GP-Full			& 8.80 & 0.68 & 0.23 & 102.0 & 199.4\\
				GP-10$^a$		& 8.90 & 1.20 & 0.28 & 12.1 & 25.6\\
			\end{tabular}
			\vspace{0.2cm}	\\
			\footnotesize{${}^a$Requires an additional $\approx 2.5$ ms for sparse approximation.} \\
			\footnotesize{${}^b$Eight outliers removed.} \vspace*{-0.3cm}
		\end{subtable}
	\end{minipage}
\end{table}
%
Table~\ref{tb:res}(b) shows the averaged simulation for different process noise realizations. The values are averaged over 200 runs, except for $T^{99.9}_c$, which is the $99.9$th percentile of all solve times. Qualitatively, the observations for the noise-free case carry over to the simulations in the presence of process noise. Most strikingly, the baseline NMPC controller displays severe constraint violations under noise. In eight cases this even causes the car to completely lose track. The runs were subsequently removed as outliers in Table~\ref{tb:res}(b). All other formulations tolerate the process noise well and achieve similar performance as in the noise-free case. The reference controller achieves slightly faster lap times than the GP-based formulations. These, however, come at the expense of higher constraint violations. Through shaping the allowed probability of violation in the chance constraints~\eqref{eq:chanceConstrFinal}, the GP-based formulations allow for a trade-off between aggressive racing and safety. \par
The simulations underline the real-time capabilities of the sparse GP-based controller. While the full GP formulation has excessive computational requirements relative to the sampling time of $T_s = 30 \text{ ms}$, the dynamic sparse formulation is solved in similar time as the baseline formulation. It does, however, require the successive update of the sparse GP formulation, which in our implementation took an additional $2.5 \text{ ms}$ on average. Note that this computation can be done directly after the previous MPC solution, whereas the MPC problem is solved after receiving a state measurement at each sample time step. The computation for the sparse approximation thus does not affect the time until an input is applied to the system, which is why we state both times separately. With $99.9\%$ of solve times below $25.6 \text{ ms}$, a computed input can be applied within the sampling time of $T_s = 30\text{ ms}$, leaving enough time for the subsequent precomputation of the sparse approximation. \par
The results demonstrate that the presented GP-based controller can significantly improve performance while maintaining safety, approaching the performance of the reference controller using the true model. They furthermore demonstrate that the controller is real-time implementable and able to tolerate process noise much better than the initial baseline controller. Overall, this indicates fitness for a hardware implementation.
%
\section{Conclusion}
\label{sc:conclusion}
%
In this paper we addressed the challenge of automatically controlling miniature race cars with an MPC approach under model inaccuracies, which can lead to dramatic failures, especially in a high performance racing environment. The proposed GP-based control approach is able to learn from model mismatch, adapt the dynamics model used for control and subsequently improve controller performance. By considering the residual model uncertainty, we can furthermore enhance constraint satisfaction and thereby safety of the vehicle. Using a dynamic sparse approximation of the GP we demonstrated the real-time capability of the resulting controller and finally showed in simulations that the GP-based approaches can significantly improve lap time and safety after learning from just one example lap.




\appendices
\section{Uncertainty propagation for nonlinear systems}
\label{app:EKF_GP}
%
Let $\mu^x_i$ and $\Sigma^x_i$ denote the mean and variance of $x_i$, respectively. Using the law of iterated expectation and the law of total variance we have
\begin{align*}
	\mu^x_{i+1} &= \E_{x_i}\left( \E_{d|x_i}\left( x_{i+1} \right) \right) \\
				&= \E_{x_i} \left( f(x_i,u_i) + B_d \mu^d(x_i,u_i) \right) \\
	\Sigma^x_{i+1} &= \E_{x_i}\left( \var_{d|x_i}\left( x_{i+1} \right) \right) + \var_{x_i} \left( \E_{d|x_i} \left( x_{i+1} \right)\right) \\
				&= \E_{x_i}\left( B_d \Sigma^d(x_i,u_i) B_d^T \right) \\
				&+ \var_{x_i} \left( f(x_i,u_i) + B_d \mu^d(x_i,u_i) \right)
\end{align*}
With a first order expansions of $f, \mu^d$ and $\Sigma^d$ around $x_i = \mu^x_i$ these can be approximated as~\cite{Candela2003}
\begin{align*}
	\mu^x_{i+1} &\approx f(\mu^x_i,u_i) + B_d \mu^d(\mu^x_i,u_i) \eqc \\
	\Sigma^x_{i+1} &\approx B_d \Sigma^d(\mu^x_i,u_i) B_d^T \\ &+ \nabla_{\!x} \tilde{f}(\mu^x_i,u_i) \Sigma^x_i {\left(\nabla_{\!x} \tilde{f}(\mu^x_i,u_i)\right)}^T
\end{align*}
with $\tilde{f}(\mu^x_i,u_i) = f(\mu^x_i,u_i) + B_d \mu^d(\mu^x_i,u_i)$.
%
%
\bibliographystyle{IEEEtran}
\bibliography{IEEEabrv,bib_OrcaGPMPC} 

\end{document}